\documentclass[preprint,12pt]{elsarticle}

\usepackage{amssymb}

\usepackage{latexsym}
\usepackage{epsfig}
\usepackage{amsmath,amsthm,amssymb}
\usepackage[linesnumbered, ruled,vlined]{algorithm2e}

\usepackage{parskip}
\usepackage{color}
\usepackage{setspace}
\usepackage{mathtools}
\usepackage{hyperref}

\def\a{\alpha}  \def\d{\delta} \def\D{\Delta}
\def\e{\epsilon}    \def\g{\gamma}
  \def\k{\kappa}
 \def\th{\theta}    
  \def\n{\nu} 
   
 \def\om{\omega}  

\def\E{\mbox{{\bf E}}}
\def\Pr{\mbox{{\bf Pr}}}

\def\uar{{\bf uar}}

\def\cG{{\cal G}}
\newcommand{\ignore}[1]{}

\newcommand{\brac}[1]{\left(#1\right)}
\newcommand{\bfrac}[2]{\left(\frac{#1}{#2}\right)}

\newtheorem{theorem}{Theorem}

\newtheorem{lemma}[theorem]{Lemma}
\newtheorem{proposition}[theorem]{Proposition}
\newtheorem{corollary}[theorem]{Corollary}

\newtheorem{definition}{Definition}

\journal{Discrete Applied Mathematics}

\begin{document}

\begin{frontmatter}



\title{Global Majority Consensus by Local Majority Polling on Graphs of a Given Degree Sequence}


\author[BU]{Mohammed Amin Abdullah}
\ead{email.mohammed@gmail.com}

\author[IC]{Moez Draief} 
\ead{m.draief@imperial.ac.uk}

\address[BU]{School of Mathematics, University of Birmingham, Edgbaston, Birmingham B15 2TT, United Kingdom}
\address[IC]{Department of Electrical and Electronic Engineering, Imperial College London, London SW7 2AZ, United Kingdom}

\begin{abstract}
Suppose in a graph $G$ vertices can be either red or blue. Let $k$ be odd. At each time step, each vertex $v$ in $G$ polls $k$ random neighbours and takes the majority colour. If it doesn't have $k$ neighbours, it simply polls all of them, or all less one if the degree of $v$ is even. We study this protocol on graphs of a given degree sequence, in the following setting: initially each vertex of $G$ is red independently with probability $\alpha < \frac{1}{2}$, and is otherwise blue. We show that if $\alpha$ is sufficiently biased, then with high probability consensus is reached on the initial global majority within $O(\log_k \log_k n)$ steps if $5 \leq k \leq d$, and $O(\log_d \log_d n)$ steps if $k > d$. Here, $d\geq 5$ is the effective minimum degree, the  smallest integer which occurs $\Theta(n)$ times in the degree sequence.  We further show that on such graphs, any local protocol in which a vertex does not change colour if all its neighbours have that same colour, takes time at least $\Omega(\log_d \log_d n)$, with high probability.  Additionally,  we demonstrate how the technique for the above sparse graphs can be applied in a straightforward manner to get bounds for the Erd\H{o}s-R\'enyi random graphs in the connected regime.
\end{abstract}

\begin{keyword}
Local majority \sep consensus  \sep social networks \sep distributed computing


\end{keyword}

\end{frontmatter}

\section{Introduction}
Let  $G=(V,E)$ be a graph where each of the $n=|V|$ vertices maintains an opinion,  which we will speak of in terms of two colours - red and blue. We are interested in distributed protocols on $G$ that can bring about consensus to a single opinion. We also desire that the opinion that was the initial majority is certain or highly likely to be the consensus reached. We make no assumptions about the properties of the colours/opinions except that vertices can distinguish between them. 

The \emph{local majority} protocol in a synchronous discrete time setting does the following: At each time step, each vertex $v$ polls all its neighbours and assumes the majority colour in the next time step. We study a generalisation of this, which we call the \emph{$k$-choice local majority protocol} where a vertex polls a random subset of $k$ neighbours, or the largest odd number it has if it doesn't have $k$. We can retrieve the local majority protocol by setting $k$ to be the maximum degree, for example. 

This can be motivated by both a prescriptive and a descriptive view. In the former, as a consensus protocol, it can be seen as a distributed co-ordination mechanism for networked systems. In the latter, it can be seen as a natural process occurring, for example in social networks where it may represent the spread of influence.

One of the simplest and most widely studied distributed consensus algorithms is the \emph{voter model} (see, e.g., \cite[ch.~14]{AlFi}). In the discrete time setting, at each time step $t$, each vertex chooses a single neighbour uniformly at random (\uar) and assumes its opinion. Thus, it can be seen as an extreme case of the $k$-choice local majority protocol, with $k=1$. The number of different opinions in the system is clearly non-increasing, and consensus is reached almost surely in finite, non-bipartite, connected graphs. Using an elegant martingale argument, \cite{Peleg} determined the probability of consensus to a particular colour. In our context this would be the sum of the degrees of vertices which started with that colour, as a fraction of the sum of degrees over all vertices. Thus, on regular graphs, for example, if the initial proportion of reds is a constant $\alpha$, the probability of a red consensus is $\alpha$. This probability is increased on non-regular graphs if the minority is ``privileged'' by sitting on high degree vertices (as in say, for example, the small proportion of high degree vertices  in a graph with power-law distribution). This motivates an alternative where the majority is certain, or highly likely, to win. 

In addition to which colour dominates, one is also interested in how long it takes to reach consensus. In the voter model, there is a duality between the voting process and multiple random walks on the graph. The time it takes for a single opinion to emerge is the same as the time it takes for $n$ independent random walks - one starting at each vertex - to coalesce into a single walk, where two or more random walks coalesce if they are on the same vertex at the same time. Thus, consensus time can be determined by studying this multiple walk process. However, the analyses of local-majority-type protocols have not been readily amenable to the established techniques for the voter model, namely, martingales and coalescing random walks. Martingales have proved elusive and the random walks duality does not readily transfer, nor is there an obvious way of altering the walks appropriately. Thus, ad-hoc techniques and approaches have been developed. 

\subsection{Related work and comparisons to our own}
In \cite{Cruise} a variant of local majority  is studied where a vertex contacts $m$ others and if $d$ of them have the same colour, the vertex subsequently assumes this colour.  They demonstrate convergence time of $O(\log n)$ and error probability -- the probability of converging on the initial minority -- decaying exponentially with $n$. However, the analysis is done only for the complete graph; our analysis of sparse graphs is a crucial difference, because the techniques employed for complete graphs do not carry through to sparse graphs, nor are they easily adapted. The error probability we give is not as strong but still strong, nevertheless. Furthermore, the convergence time we give is much smaller. 

\cite{Mossel} and \cite{Montanari}  look at local majority in different settings. In \cite{Mossel}, for $d$-regular $\lambda$-expanders they show that when there is a sufficiently large bias, there is convergence to initial majority. For the special case of $d$-regular random graphs, they show a bias equivalent to our $\alpha=1/2-1/\sqrt{d}$ is sufficient for convergence to the initial majority. Our bias condition is more demanding, i.e., whenever our condition is satisfied, theirs is too. However, their results are limited to regular graphs and they do not address timings. In \cite{Montanari} the initial setting is the same as ours; each vertex is independently red with some probability, and they analyse the bias required for almost sure convergence to the majority. However, they  study the process only on infinite regular trees.

In \cite{ColinVoting} the authors study the following protocol on random regular graphs and regular expanders:  Each vertex picks two neighbours at random, and takes the majority of these with itself. They show convergence to the initial majority in $O(\log n)$ steps with high probability, subject to sufficiently large initial bias and high enough vertex degree.

Further afield, protocols that converge to initial global majority but which are not based on local majority rules are given in \cite{Benezit} and \cite{Milan}. Whilst \cite{Benezit} applies to any connected graph, and converges to the initial global majority almost surely, the convergence time bound given in \cite{Shang} (the best general bound currently available for this protocol) is rather high, $O(n^4\log n)$. In contrast, \cite{Milan}, like \cite{Cruise}, gives a protocol which is $O(\log n)$ time and has  exponentially small error probability. However, like \cite{Cruise}, their analysis is dependent on the graph being complete.

As far as we know, our analysis is the first to demonstrate (sub)logarithmic, distributed, consensus on the initial majority with high probability (i.e., with probability tending to $1$ as $n \rightarrow \infty$),  on sparse graphs. 

\subsection{Outline of paper and some notation}
This paper concerns two graph models: The space of graphs given by a specified degree sequence, and Erd\H{o}s-R\'enyi random graphs. The former will be elaborated upon in the next section, with both definitions and structural lemmas. In section \ref{GGDS} we will give the main theorems for how the $k$-choice local majority protocol acts for this class of graphs. Subsequently, in section \ref{ERRG} we will do the same for Erd\H{o}s-R\'enyi random graphs. Finally, in section \ref{Concs}, we give a conclusion and discuss avenues for further investigation. 

For a vertex $v$ denote by $N(v)$ the neighbours of $v$ in $G$ and let $d(v)=d_v=|N(v)|$. In this paper the $\log$ function is to base $e$ when no base is stated.
  
We will denote by $G \in \mathcal{G}(n,p)$ the Erd\H{o}s-R\'enyi random graph with $n$ vertices and edge probability $p$.

For a graph $G$ and a vertex $v \in G$, let $G[v,s]$ denote the subgraph induced by the set of
vertices within a distance $s$ of $v$.

\section{Graphs of a Given Degree Sequence: Graph Model}
Let $V=[n]$ be a vertex set and define $\mathcal{G}_n(\textbf{d})$ to be the set of connected simple graphs with degree sequence $\textbf{d} =(d_1,d_2,\ldots,d_n)$, where $d_i$ is the degree of vertex $i \in V$. Clearly,  restrictions on degree sequences are required in order for the model to make sense. 
An obvious one is that the sum of the degrees in the sequence cannot be odd. Even then, not all degree sequences are \emph{graphical} and not all graphical 
sequences can produce simple graphs. Take for example the two vertices $v$ and $w$ where $d_v=3$ and $d_w=1$. In order to study this model, we restrict 
the degree sequences to those which are \emph{nice} and graphs which have nice degree sequences are termed the same. The precise definition will be given in section \ref{seca}, but basically, nice graphs are sparse, with not too many high degree vertices. They also have the property that there is a constant $0<\kappa \leq 1$ and an integer $d$ (that may grow with $n$) that occurs $\kappa n +o(n)$ times in \textbf{d}, and any integer smaller than $d$ occurs $o(n)$ times. We call $d$ the \emph{effective minimum degree}. It may be assumed where it appears below that $\textbf{d}$ is nice. 

Our analysis also requires that graphs $G$ taken from $\mathcal{G}_n(\textbf{d})$ have certain structural properties. The subset of graphs $\mathcal{G}'_n(\textbf{d})$ having these properties form a large proportion of $\mathcal{G}_n(\textbf{d})$, in fact, $|\mathcal{G}'_n(\textbf{d})|/|\mathcal{G}_n(\textbf{d})|=1-n^{-\Omega(1)}$ when $\textbf{d}$ is nice (\cite{CovDS}). We term graphs in $\mathcal{G}'_n(\textbf{d})$ for $\textbf{d}$ nice as \emph{typical}. 
Further details will be given in section \ref{cfgmod}.

For the special case of $d$-regular graphs $G$ with $d \geq 5$ constant, we can give stronger bounds on the error probability at the expense of time. To do so we require that $G$  be \emph{typical regular}, which almost all connected simple $d$-regular graphs are. More precisely, if $\mathcal{G}_n(d)$ is the set of all connected simple $d$-regular graphs on $n$ vertices, and $\mathcal{G}'_n(d) \subseteq \mathcal{G}_n(d)$ is the subset of typical regular graphs, then $|\mathcal{G}'_n(d)|/|\mathcal{G}_n(d)| \rightarrow 1$ as $n \rightarrow \infty$, as shown in \cite{CFR-Mult}.

Although this model is typically framed as a random graph, randomness here is superfluous. We assume that the graph $G$ that the protocol acts upon is from the typical subset $\mathcal{G}'_n(\textbf{d})$ of the set $\mathcal{G}_n(\textbf{d})$ of simple graphs with nice degree sequence \textbf{d}. As long as $G$ has the typical properties, the time and error probability bounds will hold. The fact that the typical subset is almost the same size as the general set is demonstrated via the configuration model. See \cite{CovDS} for a detailed explanation. 

\subsection{Assumptions about the degree sequence} \label{seca}
Let $V_j = \{i \in V: d_i = j\}$ and let $n_j = |V_j|$.
Let $\sum_{i=1}^{n}d_i = 2m$ and let
$\th=2m/n$ be the average degree. We use the notations $d_i$ and $d(i)$ for the degree of vertex $i$.

Let $0<\k\le 1$ be constant, $0<c<1/8$ be  constant and let  $\g = (\sqrt{\log n}/\th)^{1/3}$.
We suppose the degree sequence \textbf{d} satisfies
the following conditions:
\begin{description}
\item[(i)] Average degree $\th=o(\sqrt{\log n})$.
\item[(ii)]Minimum degree $\delta \geq 3$.
\item[(iii)] Let $d\geq 5$ be such that $n_d=\k n+o(n)$. We call $d$ the {\em effective minimum degree}. 
\item[(iv)]  Number of little vertices {$\displaystyle\sum_{j= \delta}^{d-1} n_j=O(n^{c(d-1)/d})$}; a vertex $v$ is \emph{little} if $d(v)\leq d-1$.
\item[(v)] Maximum degree $\Delta =O(n^{c(d-1)/d})$.
\item[(vi)] Upper tail size {$\displaystyle \sum_{j=\g\th}^\D n_j=O(\D)$}.
\end{description}

Any degree sequence with constant maximum degree, and for which $d=\d$ is nice.
The conditions hold in particular, for $d$-regular graphs, $d \ge 5, \;d=\d=o(\sqrt{\log n})$, as condition (iii)
holds with $\k=1$. The spaces of graphs we consider are somewhat more general.
The condition nice allows, for example,
bi-regular graphs where half the vertices are degree $d \ge 5$ and half of degree $\Delta=o(\sqrt{\log n})$. 

Condition (ii) ensures connectivity with high probability. Conditions (i), (iv), (v) and (vi) allow the structural properties in Lemma \ref{lightCPC} to be inferred via the configuration model, as was done in \cite{CovDS}. The effective minimum degree condition (iii),  ensures that some entry in the degree sequence occurs order $n$ times, and our analysis requires $d \geq 5$. In the conclusion, we shall discuss relaxing this condition. 
\vspace{5 mm}
\begin{definition}
A {\em nice} degree sequence \textbf{d} is one that satisfies conditions (i)--(vi) above, and we apply the
same adjective to any graph $G \in \mathcal{G}_n(\textbf{d})$ with a nice \textbf{d}.
\end{definition}

\subsection{Structural properties of $G$}\label{cfgmod}
Let $C$ be a large constant, and let
\begin{equation}\label{om}
\om=\om(n)=C\log\log n.
\end{equation}
A cycle or path  is \emph{small},
if it has at most $2\om+1$ vertices, otherwise it is \emph{large}.
Let
\begin{equation}\label{ell}
\ell =\ell(n)=B \log^2 n
\end{equation}
for some large constant $B$.
A vertex $v$ is \emph{light} if it has degree at most $\ell$, otherwise it is \emph{heavy}.
A cycle or path is \emph{light} if all  vertices are light.

Lemmas \ref{lightCPC} and \ref{dltlLB} are from \cite{CovDS}.
\vspace{5 mm}
\begin{lemma}[\cite{CovDS}]\label{lightCPC}
Let $\textbf{d}$ be a nice degree sequence and let $G$ be chosen \uar\ from 
\(\cG_n(\textbf{d})\). With probability $1-O(n^{-1/4})$,
\begin{description}
\item[(a)] No pair of vertex disjoint small light cycles is joined by a small light path.
\item[(b)] No  light vertex is in two small light cycles.
\item[(c)] No  small  cycle contains a heavy vertex or little vertex, or is connected to a heavy
or little vertex by a small path.
\item[(d)] No pair of little or heavy vertices is connected by a small  path.
\end{description}
\end{lemma}

Recall for a graph $G$ and a vertex $v \in G$, $G[v,s]$ denotes the subgraph induced by the set of
vertices within a distance $s$ of $v$.

A vertex $v$ is
\emph{$d$-tree-like to depth $h$}   if $G[v,h]$ is a $d$-regular tree, (i.e.  all vertices on levels $0,1,...,h-1$ have degree $d$).
 We choose the following value for $h$, which depends on the average degree $\th$.
\begin{equation}\label{om1}
  h =
\frac{1}{\log d} \log \bfrac{\log n}{(\log\log n) \log \th}
\end{equation}

A vertex $v$ is \emph{$d$-compliant}, if $G[v, \omega]$ is a tree,
and all vertices of $G[v,\omega]$ have degree at least $d$.
A vertex $v$ is
\emph{$d$-tree-regular}, if it is  $d$-tree-like to depth $h$, $d$-compliant to depth $\om$ and all vertices of $G[v,\omega]$ are light.
For such a vertex $v$, the first $h$ levels of the BFS tree,  really are a $d$-regular tree, and the remaining $\om-h$ levels
can be pruned to a $d$-regular tree.

\vspace{5 mm}
\begin{lemma}[\cite{CovDS}]\label{dltlLB}
Let $\textbf{d}$ be a nice degree sequence and let $G$ be chosen \uar\ from 
\(\cG(\textbf{d})\). There exists $\e>0$ constant such that 
with probability $1-O(n^{-\e})$,
\begin{description}
\item[(e)] there are $n^{1-O(1/\log\log n)}$ $d$-tree-regular vertices.
\end{description}
\end{lemma}

\vspace{5 mm}

\begin{definition}\label{typicalDef}
A \emph{typical} graph $G$ is one that is nice and also satisfies conditions (a)-(e) of Lemmas \ref{lightCPC} and \ref{dltlLB}.
\end{definition}

\vspace{5 mm}

\begin{definition}\label{typicalRegDef}
Let $L_1=\e_1\log_d n$, where $\e_1>0$ is a sufficiently small constant. A \emph{typical regular} graph $G$ of degree $d$ for some constant $d$ is one that is typical and and additionally has the following property: 
\begin{description}
\item[(f)] No pair of cycles $\mathcal{C}_1, \mathcal{C}_2$ with $|\mathcal{C}_1|, |\mathcal{C}_1|\leq 100L_1$
are within distance $100L_1$ of each other.
\end{description}
\end{definition}

From \cite{CFR-Mult} we have the following,
\vspace{5 mm}
\begin{lemma}[\cite{CFR-Mult}]
Let $G$ be chosen \uar\ from $\mathcal{G}_n(d)$, the set of all simple connected $d$-regular graphs. With probability tending to $1$ as $n \rightarrow \infty $, no pair of cycles $\mathcal{C}_1, \mathcal{C}_2$ with $|\mathcal{C}_1|, |\mathcal{C}_1|\leq 100L_1$
are within distance $100L_1$ of each other.
\end{lemma}

Thus, we see that a fraction $1-n^{-\Omega(1)}$ of graphs \(G \in \cG_n(\textbf{d})\) are typical and a fraction $1-o(1)$ of \(G \in \cG_n(d)\), the set of $n$-vertex simple connected $d$-regular graphs, are typical regular. 

\section{Graphs of a Given Degree Sequence: Results}\label{GGDS}

For all the results we give - both lower and upper bounds - we assume the same initialisation of the colours of the vertices: initially, each vertex of $G$ is red independently with probability $\alpha \in (0, \frac{1}{2})$, and is otherwise blue. A simple Chernoff bound argument gives  the following:
\\

\begin{proposition}\label{BlueProp}
For come constant $c>0$, blue is the initial majority with probability at least $1-e^{-cn}$. 
\end{proposition}

\subsection{The lower bound}
\begin{definition}
We call a protocol $\mathcal{P}$ \emph{local-stable} (\emph{LS}) if it has the following properties
\begin{description}
\item[(i)] It is \emph{local} meaning that a vertex $v$ can only directly exchange information with its neighbours
\item[(ii)] It is \emph{stable}, meaning that if a vertex $v$ and all its neighbours $N(v)$ are the same colour, then under $\mathcal{P}$, $v$ will not change colour in the next step
\end{description}
\end{definition}

\vspace{5 mm}

\begin{theorem}\label{LB}
Suppose $G \in \mathcal{G}_n(\textbf{d})$ is typical with effective minimum degree $d$. Suppose initially each vertex of $G$ is red independently with probability $\alpha \in (0,\frac{1}{2})$, and is blue otherwise. For any local-stable protocol $\mathcal{P}$, the following holds:  With probability $1-e^{-\Omega(n^{1-o(1)})}$, at time step $(1-o(1))\log_d\log_d n$, $\mathcal{P}$ will not have reached consensus on the initial majority. 
\end{theorem}
\begin{proof}
Suppose $v$ is $d$-tree-regular. The probability all vertices in $G[v,h]$ are initially red is $\alpha^K$ where $K = 1+\frac{d}{d-2}[(d-1)^h-1]$.
Since $G$ is typical, there exist $n^{1-O(1/\log\log n)}$ 
$d$-tree-regular vertices (Definition \ref{typicalDef}) so there will be $\Omega\left(n^{1-O(1/\log\log n)}/\ell^{2\omega}\right)$ non-intersecting $d$-regular trees of depth $h$ (recall the definition of $\ell$ from equation \eqref{ell}). Therefore, the probability that at least one of these is initially all red is at least 
\begin{eqnarray}
1-(1-\alpha^{K})^{\Omega\left(n^{1-O(1/\log\log n)}/\ell^{2\omega}\right)} &\geq& 1-\exp\left\{-\Omega\left(\frac{\alpha^{K}n^{1-O(1/\log\log n)}}{\ell^{2\omega}}\right)\right\} \nonumber\\
&\geq& 1-\exp\left\{-\Omega\left(\frac{n^{1-O(1/\log\log n)}}{c^{3d^h}\ell^{2\omega}}\right)\right\},\label{fnsnf}
\end{eqnarray}
where $c=\frac{1}{\alpha}>2$. The logarithm of the bracketed expression in \eqref{fnsnf} is 
\begin{align*}
&\brac{1-O\brac{\frac{1}{\log\log n}}}\log n-3d^h\log c-2\omega \log \ell\\
&=\brac{1-O\brac{\frac{1}{\log\log n}}}\log n-3(\log c)\frac{\log n}{(\log\log n) \log \th}-4C(\log \log n)^2-2C(\log B)\log\log n\\
&=\brac{1-O\brac{\frac{1}{\log\log n}}}\log n.
\end{align*}
Thus \eqref{fnsnf} is $1-e^{-\Omega(n^{1-o(1)})}$.

By the locality and stability conditions, $G[v,h]$ being initially all red means it requires at least $h$ time steps until $v$ can become blue. The theorem follows by Proposition \ref{BlueProp}. 
\end{proof}

\subsection{The upper bounds}

We first formally define the protocol. Time is indexed by the non-negative integers $t=0, 1, 2, \ldots$. For a given graph $G$, let $X_t^{\mathcal{P}}(v)$ be the indicator function for vertex $v$ being blue at time $t$ under protocol $\mathcal{P}$, i.e., $X_t^{\mathcal{P}}(v)=1$ iff $v$ is blue at time step $t$ when running protocol $\mathcal{P}$. For a positive integer $k$ and a vertex $v$, define $v(k)=\min\left\{k,2\left\lfloor\frac{d(v)-1}{2}\right\rfloor+1\right\}$. Thus, $v(k)$ is the minimum of $k$ and the largest odd number not greater than $d(v)$. Below we assume $k$ is odd.

\vspace{5 mm}

\begin{definition}[$k$-choice Local Majority Protocol $\mathcal{MP}^k$]
At time step $t+1$, each vertex $v \in V$ randomly picks a $v(k)$-subset $N_v(t+1)$ uniformly from the set of $\binom{d(v)}{v(k)}$ possible subsets. $v$ then assumes at time $t+1$ the majority colour at time $t$ of the vertices in $N_v(t+1)$. More formally, 
\[
X_{t+1}^{\mathcal{MP}^k}(v)=\mathbf{1}_{\left\{\brac{\sum_{w \in N_v(t+1)}X_{t}^{\mathcal{MP}^k}(w)} > \frac{v(k)}{2}\right\}}.
\]
\end{definition}

For real $\alpha$ and integer $\nu\geq 2$ define
\begin{equation}
f(\nu, \alpha)= \left[\left(1+\frac{1}{\sqrt{2\nu}}\right)2\right]^{\frac{1}{\nu-1}}4\alpha(1-\alpha). \label{condFunc}
\end{equation}

In this section we shall prove the following:

\vspace{5 mm}

\begin{theorem}\label{MPUB}
Let  $G \in \mathcal{G}_n(\textbf{d})$ be typical with  effective minimum degree $d$.

Suppose initially each vertex of $G$ is red independently with probability $\alpha \in (0, \frac{1}{2})$, and is otherwise blue. 

\textbf{case $5\leq k \leq d$:} Let $\nu=\frac{k-1}{2}$ and suppose $f\brac{\nu, \alpha}<\beta$ for some constant $\beta<1$. 
Then for any constant $\varepsilon>0$, with probability $1-n^{-\Omega((\log n)^{\varepsilon/2})}$,  $\mathcal{MP}^{k}$ will have reached consensus on the initial majority by time step $\frac{1+\varepsilon}{\log_k \nu}\log_k\log_k n$.

\textbf{case $k>d$:} Let $d^o$ be the largest odd number not greater than $d$, let $\nu=\frac{d^o-1}{2}$ and  suppose $f\brac{\nu, \alpha}<\beta$ for some constant $\beta<1$. Then for any constant $\varepsilon>0$, with probability $1-n^{-\Omega((\log n)^{\varepsilon/2})}$,  $\mathcal{MP}^{k}$ will have reached consensus on the initial majority by time step $\frac{1+\varepsilon}{\log_{d^o} \nu}\log_{d^o}\log_{d^o} n$.
\end{theorem}

Observe that $k$ is not required to be a constant; it is allowed to be a function of $n$ which goes to infinity. 

Consequent of Theorem \ref{LB} and Theorem \ref{MPUB}, $\mathcal{MP}^k$ is asymptotically optimal when $k=d$ (in fact, when $d/k=O(1))$. Note $\frac{1+\varepsilon}{\log_k \nu}$ is at most $3(1+\varepsilon)$ by the assumption $k\geq 5$, and it is $(1+o(1))(1+\varepsilon)$ if $k \rightarrow \infty$ with $n$. Thus, when $k=d\rightarrow \infty$ the upper bound is within factor $1+\varepsilon$ of the lower bound, for arbitrarily small constant $\varepsilon>0$.

An immediate corollary concerns random regular graphs:
\\

\begin{corollary}
Let $G \in \mathcal{G}_n(d)$ be drawn uniformly at random from the set of simple, connected $d$-regular graphs on $n$ vertices, where $5 \leq d =o(\sqrt{\log n})$. Let $5 \leq k \leq d$ be odd. 
 Suppose initially each vertex of $G$ is red independently with probability $\alpha \in (0,\frac{1}{2})$, and is otherwise blue. If $f\brac{\frac{k-1}{2}, \alpha}<\beta$ for some constant $\beta<1$, then with high probability,  $\mathcal{MP}^{k}$ will reach consensus on the initial majority within $\Theta(\log_k \log_k n)$ steps. 
\end{corollary}

For regular graphs, we can get a stronger error probability at the expense of time. For simplicty, we will only state it for the case $k=d$,  with $d$ odd.
\\

\begin{theorem}\label{regThm}
Suppose $G \in \mathcal{G}_n(d)$ is typical regular with (odd) effective minimum degree $d$ . Suppose initially each vertex of $G$ is red independently with probability $\alpha \in (0,\frac{1}{2})$, and is otherwise blue. If $f\brac{\frac{d-1}{2}, \alpha}<\beta$ for some constant $\beta<1$, then for some constants $c>1$ and $0<\varepsilon<1$, with probability $1-O\brac{c^{-n^{\varepsilon}}}$, by time $O(\log n)$ $\mathcal{MP}^{d}$ will have reached consensus on the initial majority. 
\end{theorem}

We give a brief outline of the key ideas that we use. Core to our analysis is the tree subgraph. Consider a tree $\mathcal{T}(v)$ with root $v$. Suppose for simplicity that all non-leaf vertices have degree $d+1$, where $d\geq 5$ is odd,  and that the tree is of depth $h$. Set $k=d$. 
Instead of picking from all neighbours, we consider a slightly different protocol where non-leaf vertices only poll their children and take the majority colour from that. 
At $t=0$, each vertex is red independently with probability $\alpha<1/2$. Now consider a vertex $x$ at depth $h-1$, meaning $x$ is parent to leaves. At time $t=1$, $x$ polls its children in  $\mathcal{T}(v)$ and takes the majority colour. $x$'s own colour at time $t=0$ is irrelevant. At time $t=1$, $x$ will be red with probability $p_1=\Pr(\text{Bin}(d,\alpha)>d/2)$, which is  less than $p_0=\alpha$. All the other vertices at depth $h-1$ have the same probability of being red. Furthermore, since vertices only poll their children, those at depth $h-1$ are independent of each other. Thus, from the point of view of vertices at depth $h-2$, at time $t=2$, they poll children, each of which is independently red with probability $p_1$. Hence, at $t=2$, a depth $h-2$ vertex is red with probability  $p_2=\Pr(\text{Bin}(d,p_1)>d/2)$ which is less than $p_1$. Continuing in this manner, the sequence $p_0, p_1, \ldots$ decays exponentially, and by time $t=h$, the root vertex will be sampling children which have a very low probability of being red such that taking the union bound over all such ``tree-like'' vertices $v$, means with high probability none of them are red at time $h$. 

In our analysis, we use this tree local subgraph and analyse a modification of the majority protocol where vertices in the tree make the conservative assumption that their parents are red. This can only do worse in terms of error probability so provides a valid bound. Of course not all vertices in the graphs under consideration are regular and locally tree-like, and we make use of results on structural properties of these graphs to handle those vertices close to cycles or which have small degree vertices in the locality.

For any $v \in V$, there is some integer $s \geq 0$ such that $\mathcal{T} = G[v,s]$ is a tree rooted at $v$. We define the \emph{modified majority protocol} $\mathcal{MMP}^k(v,s)$ with respect to a vertex $v \in V$. Recall $X_t^{\mathcal{MP}^k}(x)=1$ if, under $\mathcal{MP}^k$, $x$ is blue at time $t$ and $0$ if it is red. Let $X_t^{\mathcal{MMP}^k(v,s)}(x)$ be the same for $\mathcal{MMP}^k(v,s)$. 

\vspace{5 mm}

\begin{definition}[$k$-choice Modified Majority Protocol $\mathcal{MMP}^k(v,s)$]
Let $\mathcal{T} = G[v,s]$. At time step $t+1$, each vertex $x \in V$ randomly picks a $x(k)$-subset $N_x(t+1)$ uniformly from the set of $\binom{d(x)}{x(k)}$ possible subsets. If $x \notin \mathcal{T}$ then $x$ becomes at time $t+1$ the majority colour at time $t$ of the vertices in $N_x(t+1)$. More formally, 
\[
X_{t+1}^{\mathcal{MMP}^k(v,s)}(x)=\mathbf{1}_{\left\{\brac{\sum_{y \in N_x(t+1)}X_{t}^{\mathcal{MP}^k}(y)} > x(k)/2\right\}}.
\]

If $x \in \mathcal{T}$ then denote by $\text{Par}(x)$ the parent of $x$ in $\mathcal{T}$. At time $t+1$, $x$ becomes the majority colour at time $t$ of the vertices in $N_x(t+1)$, with the added assumption that $\text{Par}(x)$ was red at time $t$. More formally, 
\[
X_{t+1}^{\mathcal{MMP}^k(v,s)}(x)=\mathbf{1}_{\left\{\brac{\sum_{y \in N_x(t+1) \setminus \{\text{Par}(x)\}}X_{t}^{\mathcal{MMP}^k(v,s)}(y)} > x(k)/2\right\}}.
\]
\end{definition}
 
Thus, $\mathcal{MMP}$ is the same as $\mathcal{MP}$ except that vertices in $\mathcal{T}$ effectively make the conservative assumption that a parent is picked and it is red. This will help in getting an upper bound on the probability of red.

In the next lemma, we couple $\mathcal{MMP}^k$ and $\mathcal{MP}^k$ to show how the former bounds the latter. In order to do so, we make use of the fact that the randomness of the system is not affected by the actions of the protocols. To reiterate, given a graph $G$, there are two sources of randomness; there is the initial random assignment of colours, and there is the sequence of choices of neighbours $(N_v(t))_{t=1}^{\infty}$ each vertex $v$ makes. Thus, the Cartesian product $\Omega$ of the possible initial colourings with each of the infinite sequences of neighbour choices each vertex creates a probability space, where an element $\sigma \in \Omega$ is a particular initial colouring of the vertices and, for each vertex $v$, a particular infinite sequence of neighbour choices made by $v$. The next lemma compares $\mathcal{MMP}^k$ and $\mathcal{MP}^k$ under the same $\sigma \in \Omega$.

\vspace{5 mm}

\begin{lemma}\label{MPUBLemma1}
For a graph $G$, a vertex $v \in G$ and $s \geq 1$,  suppose $\mathcal{T} = G[v,s]$ is a tree. Fix $\sigma \in \Omega$ and consider $\mathcal{MP}^k$ and $\mathcal{MMP}^k(v,s)$ under this $\sigma$. For all $t\geq 0$,  we have  $X_{t}^{\mathcal{MMP}^k(v,s)}(x) \leq X_{t}^{\mathcal{MP}^k}(x)$ for every vertex $x \in G$.
\end{lemma}
\begin{proof}
Below we will forgo `$k$' and`$(v,s)$' in the superscripts. $N_x(t)$ the set of neighbours chosen by vertex $x$ at time $t$ under $\sigma$.
 
We argue by induction on $t$. Clearly $X_{0}^{\mathcal{MMP}}(x) = X_{0}^{\mathcal{MP}}(x)$ for every $x$. Suppose $X_{t}^{\mathcal{MMP}}(x) \leq X_{t}^{\mathcal{MP}}(x)$ for every $x$. If $x \notin \mathcal{T}$ then 
\[
X_{t+1}^{\mathcal{MMP}}(x)=\mathbf{1}_{\{\sum_{y \in N_x(t+1)}X_{t}^{\mathcal{MMP}}(y) > \frac{x(k)}{2}\}} \leq \mathbf{1}_{\{\sum_{y \in N_x(t+1)}X_{t}^{\mathcal{MP}}(y) > \frac{x(k)}{2}\}}=X_{t+1}^{\mathcal{MP}}(x)
\]
If $x \in \mathcal{T}$ then   
\[
\sum_{y \in N_x(t+1) \setminus \{\text{Par}(x)\}}X_{t}^{\mathcal{MMP}}(y) \leq \sum_{y \in N_x(t+1) \setminus \{\text{Par}(x)\}}X_{t}^{\mathcal{MP}}(y) \leq \sum_{y \in N_x(t+1)}X_{t}^{\mathcal{MP}}(y),
\]
so
\begin{eqnarray*}
X_{t+1}^{\mathcal{MMP}}(x)&=&\mathbf{1}_{\{\sum_{y \in N_x(t+1) \setminus \{\text{Par}(x)\}}X_{t}^{\mathcal{MMP}}(y) > \frac{x(k)}{2}\}}\\
&\leq& \mathbf{1}_{\{\sum_{y \in N_x(t+1)}X_{t}^{\mathcal{MP}}(y) > \frac{x(k)}{2}\}}\\
 &=& X_{t+1}^{\mathcal{MP}}(x).
\end{eqnarray*}
\end{proof}

\begin{corollary}\label{probCor}
Suppose $G[v,s]$ is a tree. Then 
\[
\Pr(X_{s}^{\mathcal{MP}^k}(v)=0) \leq \Pr(X_{s}^{\mathcal{MMP}^k(v,s)}(v)=0).
\]
\end{corollary}

\vspace{3 mm}

Let
\[
\omega'(n,k)= \log_k \log_k n.
\]
\vspace{3 mm}

We re-iterate that both $k$ and $d$ may be functions of $n$ that go to infinity:
\\

\begin{lemma}\label{problemma}
Let  $G \in \mathcal{G}_n(\textbf{d})$ be typical with  effective minimum degree $d$. Let $k$ be odd with $5 \leq k \leq d$ and let $\nu=\frac{k-1}{2}$. Let $\varepsilon$ be any positive constant and let $A=\frac{1+\varepsilon}{log_k \nu}$. Let $\omega'=\omega'(n,k)$ and suppose $G[v, A\omega']$ is a tree with each non-leaf vertex having degree at least $d$. Suppose each vertex in $G$ is initially red independently with probability $\a \in (0,\frac{1}{2})$, and is otherwise blue. If $f(\n, \alpha)<\beta$ for some constant $\beta<1$, then 
\[
\Pr(X_{A\omega'}^{\mathcal{MP}^k}(v)=0) = n^{-\Omega((\log n)^{\varepsilon/2})}.
\]
\end{lemma}
\begin{proof}
We bound $\Pr(X_{A\omega'}^{\mathcal{MMP}^k(v,s)}(v)=0)$ and use Corollary \ref{probCor}. For convenience, we sometimes forgo notation that is obvious. 

For a vertex $x$, let $d(v,x)$ be the distance of $x$ from $v$ in $G$, and for $t \geq 0$, let $p_t(x)=\Pr(X_{t}^{\mathcal{MMP}}(x)=0)$. For any vertex $x$, $p_0(x)=\Pr(X_{0}^{\mathcal{MMP}}(x)=0)=\alpha$, and in particular, this holds for any vertex $x$ such that $d(v,x)=A\omega'$.

Now for $x$ with $d(v,x)=A\omega'-1$, 
\begin{equation}
p_1(x) =\Pr\left(\sum_{y \in N_x(1) \setminus \{\text{Par}(x)\}}X_{0}^{\mathcal{MMP}}(y) \leq \nu\right) \leq \sum_{i=\nu}^{2\nu}\binom{2\nu}{i}\alpha^i(1-\alpha)^{2\nu-i}. \label{eqp1}
\end{equation}

Since $\alpha<\frac{1}{2}$,
\begin{eqnarray*}
\sum_{i=\nu}^{2\nu}\binom{2\nu}{i}\alpha^i(1-\alpha)^{2\nu-i} &\leq& \alpha^{\nu}(1-\alpha)^{\nu}\sum_{i=\nu}^{2\nu}\binom{2\nu}{i}\\
&=&\alpha^\nu(1-\alpha)^\nu \brac{\frac{1}{2}2^{2\nu}+\frac{1}{2}\binom{2\nu}{\nu}},
\end{eqnarray*}
and using the inequality $\binom{2n}{n}\leq \frac{2^{2n}}{\sqrt{2n}}$, we have  $p_1(x) \leq \frac{1}{2}(1+\frac{1}{\sqrt{2\nu}})(4\alpha(1-\alpha))^\nu$. 

Assume for $t< A\omega'$ and all $x$ such that $d(v,x)=A\omega'-t$,
\[
p_{t}(x) \leq	\frac{1}{4}\left[\left(1+\frac{1}{\sqrt{2\nu}}\right)2\right]^{\sum_{i=0}^{t-1}\nu^i}\left(4\alpha(1-\alpha)\right)^{\nu^t}, 
\]
and define $p_t$ to be the RHS of the above inequality. 

Then for $t+1$ and all $x$ such that $d(v,x)=A\omega'-t-1$,
\begin{eqnarray*}
p_{t+1}(x)&\leq&\sum_{i=\nu}^{2\nu}\binom{2\nu}{i}p_t^i(1-p_t)^{2\nu-i}\\
&\leq& \frac{1}{2}\left(1+\frac{1}{\sqrt{2\nu}}\right)\left(4p_t(1-p_t)\right)^\nu\\
&\leq& \frac{1}{2}\left(1+\frac{1}{\sqrt{2\nu}}\right)\left(4p_t\right)^\nu\\
&\leq& \frac{1}{2}\left(1+\frac{1}{\sqrt{2\nu}}\right)\left(4\frac{1}{4}\left[\left(1+\frac{1}{\sqrt{2\nu}}\right)2\right]^{\sum_{i=0}^{t-1}\nu^i}\left(4\alpha(1-\alpha)\right)^{\nu^t}\right)^\nu\\
&=& \frac{1}{4}\left[\left(1+\frac{1}{\sqrt{2\nu}}\right)2\right]^{\sum_{i=0}^{t}\nu^i}\left(4\alpha(1-\alpha)\right)^{\nu^{t+1}}\\
&=&p_{t+1}.
\end{eqnarray*}

Hence for any $t \leq A\omega'$ and all $x$ such that $d(v,x)=A\omega'-t$,
\[
p_t(x) \leq \frac{1}{4}\left(\left[\left(1+\frac{1}{\sqrt{2\nu}}\right)2\right]^{\frac{1}{\nu-1}}4\alpha(1-\alpha)\right)^{\nu^t}=\frac{1}{4}\brac{f(\nu, \alpha)}^{\nu^t}.
\] 
In particular, when $t=A\omega'$, $\nu^t= \nu^{A\omega'}=(\log_k n)^{A\log_k \nu}$, and by the condition $f(\n, \a)<\beta$,
\begin{eqnarray}
\Pr(X_{A\omega'}^{\mathcal{MMP}}(v)=0) &\leq&  \frac{1}{4}\beta^{(\log_k n)^{A\log_k \nu}}\nonumber \\
&=&\frac{1}{4} n^{-\log_k(\frac{1}{\beta})(\log_k n)^{A\log_k \nu-1}}.\label{2sn}
\end{eqnarray}

Now $\log_k \nu=\log _k(\frac{k-1}{2})>0.43$ since we assume $k \geq 5$.  
Hence, if $A=\frac{1+\varepsilon}{\log_k \nu}$ where $\varepsilon>0$ is a constant, then the positive part of the exponent in \eqref{2sn} is 
\begin{eqnarray*}
\log_k\brac{\frac{1}{\beta}}(\log_k n)^{\varepsilon}&=& \log\brac{\frac{1}{\beta}}\frac{(\log n)^{\varepsilon}}{(\log k)^{1+\varepsilon}}\\ &=&\Omega\brac{\frac{(\log n)^{\varepsilon}}{(\log \log n)^{1+\varepsilon}}}\\
&=&\Omega((\log n)^{\varepsilon/2}), 
\end{eqnarray*}
where the second equality holds because $k \leq d=o(\sqrt{\log n})$. Thus, \eqref{2sn} is $n^{-\Omega((\log n)^{\varepsilon/2})}$ and applying Corollary \ref{probCor} completes the proof.
\end{proof}

The above lemma deals with vertices $v$ for which $G[v, A\omega']$ is a tree and each non-leaf vertex in this tree has degree at least $d$.  We are left to deal with vertices $v$ for which $G[v, A\omega']$ contains a cycle or a non-leaf vertex with degree less than $d$. 

\vspace{5 mm}

\begin{lemma}
\label{nontreelemma}
 Lemma \ref{problemma} holds for $G[v,A\omega']$ when $G[v, A\omega']$ contains a cycle or a non-leaf vertex with degree less than $d$.
\end{lemma}
\begin{proof}
Oberve first of all that $C$ in the definition of $\omega$ is an arbitrarily large constant. Hence, we can assume $A\omega' \leq \omega$. 

In Lemma \ref{lightCPC} \textbf{(c)} says all vertices within distance $2\omega+1$ of a small cycle $\mathcal{C}_1$ are light, so any other small cycle $\mathcal{C}_2$ within $2\omega+1$ either connects to $\mathcal{C}_1$ via a small light path or $\mathcal{C}_1$ and $\mathcal{C}_2$ intersect. The former case is precluded by \textbf{(a)} and the latter by \textbf{(b)} of the same lemma. Therefore, no pair of small cycles is within distance $2\omega+1$ of each other. Hence, if for some $v$, $G[v,A\omega']$ is not a tree, then for some (unique) cycle $\mathcal{C}$, $v$ is either on $\mathcal{C}$ or there is a unique small path from $v$ to $\mathcal{C}$. Consider the latter case. Suppose $x \in N(v)$ is on the small path. We may assume $x$ always to be red and since $v$ has degree at least $d$ (by Lemma \ref{lightCPC} \textbf{(c)}), the bound in \eqref{2sn} holds. 

Now suppose $v$ is on the cycle $\mathcal{C}$. Suppose $\{x,y\}$ are the neighbours of $v$ on $\mathcal{C}$. Each vertex $u \in N(v) \setminus \{x,y\} $ has degree at least $d$ (by Lemma \ref{lightCPC} \textbf{(c)}) and can assume $v$ is always red to get the bound in \eqref{2sn}. Thus by time $A\omega'$ all of $N(v) \setminus \{x,y\}$ will be blue and in the next step they will out-vote $\{x,y\}$ if $d-2>2$. 

We now deal with little vertices. By Lemma \ref{lightCPC} \textbf{(d)} there can be at most one vertex $u \in G[v,A\omega']$ such that $d(u)<d$. If $u \equiv v$ then by the above argument all of $N(v)$ will be blue by time $A\omega'$ and so $v$ will be in the next step. Otherwise, there is a unique path from $v$ to $u$ which can be cut off, and by the above argument the bound holds. 
\end{proof}

We now address the case $k>d^o$ where $d^o$ is the largest odd number not greater than $d$. Intuitively, it would  seem that having a higher degree could only help the root vertex become the the initial majority, since a larger number of child vertices are being sampled. Indeed, the following  formally justifies this intuition (proof in appendix):
\\

\begin{proposition}\label{binprop}
Let $N$ be a natural number and $p \in (0,\frac{1}{2})$. Then 
\[
\Pr\brac{\operatorname{Bin}(2N,p)\geq N} \geq \Pr\brac{\operatorname{Bin}(2N+2,p)\geq N+1}. 
\]
\end{proposition}

As such, if $\mathcal{MMP}^k$ is run with $k>d^o$, the probability the root is red is at most the probability for $\mathcal{MMP}^{d^o}$. Thus, for $k>d^o$, we may apply Lemma \ref{problemma} setting $k=d^o$ and arrive at a correct conclusion.   Why then not derive a similar result directly for $k>d^o$?  On the assumption that the degrees of vertices in the tree are at least $k+1$, and $k =O(\log n)$, one may do so, getting a blue root vertex within  $O(\log_k \log_k n)$ steps.  However, the result would not hold for all vertices. The effective minimum degree $d$ is a natural upper barrier since it defines the structural result given in Lemma \ref{dltlLB}, which we apply to get the lower bound. Thus, even if some vertices become blue more quickly in time $O(\log_k \log_k n)$, there will be some which require $\Omega(\log_{d^o} \log_{d^o} n)$ steps. 

\begin{proof}[\textbf{Proof of Theorem \ref{MPUB}}]
For $5 \leq k \leq d$, using Lemmas \ref{problemma} and \ref{nontreelemma}, apply a union bound to all $n$ vertices in $G$. The case $k>d^o$ follows by the above discussion. Finally, we apply Proposition \ref{BlueProp}. 
\end{proof}

\begin{proof}[\textbf{Proof of Theorem \ref{regThm}}]
This follows by the same reasoning as the proof of Theorem \ref{MPUB}, except that now we use trees of depth $L_1$, giving a time $L_1$ and error bound
\eqref{2sn} of 
\begin{equation}
\Pr(X_{L_1}^{\mathcal{MMP}}(v)=0) \leq  \frac{1}{4}\beta^{\nu^{L_1}}=O(c^{-n^\varepsilon}) 
\end{equation}
for some pair of constants $c>1$ and $0<\varepsilon<1$.
\end{proof}



\section{Erd\H{o}s-R\'enyi Random Graphs}\label{ERRG}
We study $\mathcal{G}_n(p)$ for $p=\frac{c\log n}{n}$, where $c>2+\epsilon$ and $\epsilon>0$ is a constant.  We shall take $k$ to be an odd constant, though it should not be too difficult to extend the results for values which go to infinity with $n$.

Theorem \ref{ERThSmall} below gives a probability for a trial whereby a graph $G$ from $\mathcal{G}(n,p)$ is picked then the protocol run on it, i.e., there is randomness in the actual graph the protocol is run on. This is in contrast to Theorems \ref{LB},  \ref{MPUB} and \ref{regThm} where there was no randomness in picking the graph; it merely had to be a graph from the typical set. When $p=1$, however, $\mathcal{G}(n,p)$ will be the complete graph with probability $1$. In this case, it is interesting to compare this result with those given in \cite{Cruise}; the error probability in Theorem \ref{ERThSmall} is not as strong as the exponentially small one given in \cite{Cruise} but the convergence time is $O(\log\log n)$ compared to their $O(\log n)$.

\vspace{5 mm}

\begin{theorem}\label{ERThSmall}
Let $\epsilon>0$ be a constant and $p \geq (2+\epsilon)\log n/n$.  Pick a graph $G \in \mathcal{G}(n,p)$ and suppose each vertex in $G$ is initially red independently with probability $\a \in (0,\frac{1}{2})$, and is otherwise blue. 

Let $k\geq 5$ be an odd constant and let $\n=\frac{k-1}{2}$. If $f(\n, \alpha)<\beta$ for some constant $\beta<1$, then for any positive constant $\varepsilon$, the following holds with probability $1-n^{-\Omega(1)}$: By time $\frac{1+\varepsilon}{\log _k\nu}\log_k\log_k n$ $\mathcal{MP}^{k}$ will have reached consensus on the initial majority.
\end{theorem}
\begin{proof}
We use the  Chernoff bound given in Theorem \ref{Chernoff}, in the appendix. If $X$ is the degree of a particular vertex $v$, $X \sim \text{Bin}(n-1, p)$. Let $\epsilon_1$ be a constant such that $1>\epsilon_1>\sqrt{2/(2+\epsilon)}$. Then 
\[
\Pr(X<(1-\epsilon_1)(2+\epsilon)\log n) \leq \exp\left(-\frac{\epsilon_1^2}{2}(2+\epsilon)\log n\right)=n^{-(1+c_1)}, 
\]
for some constant $c_1>0$. Taking the union bound over all vertices means that with probability at least $1-n^{-c_1}$ each vertex has degree at least $c_2\log n$ for some constant $c_2>0$. 

Our approach is as follows: We fix a vertex $v$ and choose a positive integer $T$.  We work backwards in time from $t=T$ looking at the $k$ vertices polled by $v$ at time $T$, then look at the $k^2$ vertices polled by those vertices at time $T-1$ and so on. We shall show that if If $T$ is not too large, the polling decisions will have a tree structure, or will be a tree plus one extra edge. We shall see that we can choose $T$ to be $A\omega'$, where $A=\frac{1+\varepsilon}{\log _k\nu}$ and $\omega'=\log_k \log_k n$. We then calculate the probability of $v$ being blue at time $A\omega'$ based on the same principle as the previous section.

For a vertex $v$, denote the set of $k$ vertices chosen by $v$ at time $t$ by $N_v(t)$. Consider a ``root'' vertex $v$. We build a (multi)graph $\mathcal{T}=\mathcal{T}(v)$ with the following algorithm.  The data structure $map$  associates a vertex with a value.  $E(\mathcal{T})$ is the edge multiset of $\mathcal{T}$. The vertex set $V(\mathcal{T})$ is not a multiset. 

\begin{algorithm}
$\mathcal{T} \longleftarrow \{v\}$ \;
$\text{map}(v) \longleftarrow 0$ \;
\For{$i \leftarrow 0$ \KwTo $A\omega'-1$}{
	\ForEach{$x\in V(\mathcal{T})$}{
		\If{$\text{map}(x) = i$}{
			$V(\mathcal{T}) \longleftarrow V(\mathcal{T}) \cup N_x(A\omega'-i)$\;
			\ForEach{$y \in N_x(A\omega'-i)$}{
				$\text{map}(y) \longleftarrow i+1$\;
				$E(\mathcal{T})\longleftarrow E(\mathcal{T}) \cup (x,y)$\;
			}
		}
	}
}
\caption{$\mathcal{T}$-BUILD}
\end{algorithm}

Before two cycles have formed, no vertex has exposed more than $2k+2$ edges. Given that $k$ is a constant and each vertex has degree $\Omega(\log n)$, this exposure is negligible. Since $|V(\mathcal{T})| \leq k^{A\omega'+1}$,  the probability, at any particular step before two cycles are formed, of connecting to $\mathcal{T}$ is $O\brac{k^{A\omega'+1}/n}$. Therefore, the probability that at any point in its construction $\mathcal{T}$ is picked twice is at most 
\[
O\brac{\frac{k^{2(A\omega'+1)}}{n^2}} \times \binom{k^{A\omega'+1}}{2}=O\brac{\frac{k^{4(A\omega'+1)}}{n^2}}= O\brac{\frac{(\log_k n)^{4A}}{n^2}}
\]    
where we have used the assumption that $k$ is constant in the last equality. Taking the union bound over all $n$ originating vertices, we see that with probability $1-n^{-\Omega(1)}$, $\mathcal{T}(v)$ has at most one cycle for every $v$ in $G$.

We can now apply the same reasoning as in the proof of Theorem \ref{MPUB}. If for a $v$ $\mathcal{T}(v)$ has a cycle $\mathcal{C}$ , then if $v$ is not on $\mathcal{C}$ we can cut off the (unique) branch containing $\mathcal{C}$ and the bounds in \eqref{2sn} holds, as per the proof of Lemma \ref{nontreelemma}. If $v$ is on $\mathcal{C}$, then since $k-2>2$, the two neighbours  of $v$ on $\mathcal{C}$ are out-voted by those not on $\mathcal{C}$, again, as per the proof of Lemma \ref{nontreelemma}.

\end{proof}

\section{Conclusion and Further Work}\label{Concs}
We have studied a variant of the local majority protocol on two types of graphs: those with a prescribed degree sequence, and Erd\H{o}s-R\'enyi random graphs in the connected regime. We have shown that when each vertex starts red independently with probability less than half, then the process will converge to the initial global majority with high probability, and will do so in sublogarithmic time. We have also demonstrated lower bound for convergence time that is within a small factor of the upper bound for the prescribed degree sequence graphs.

There are a number of possible directions for further research. Recall that nice degree sequences required the effective minimum degree $d$ to be at least $5$. Since we require minimum degree $3$ to ensure connectivity, it would appear to be a suitable target for reducing the bound on $d$. However, if there are $\Omega(n)$ vertices with degree $3$, then there could be cycles of degree $3$ vertices with some of those cycles being entirely the minority colour. Since they can't change colour, consensus can't be reached on the initial global majority. In this case, an alternate protocol such as choosing neighbours uniformly at random with replacement could be an interesting line of investigation. For the protocol of this paper, $d=4$ remains a valid target. 

Another direction is studying the process on other graph models. In particular, those  exhibiting inhomogeneous structure such as power-law distributions (see, e.g., \cite{Remco} for an up to date treatment of such models).

Lastly, one may consider a setting in which the initial distribution of colours is not necessarily random, and investigate which distributions allow for a convergence to the initial global majority.  


\section{Acknowledgement}
We thank Colin Cooper for discussions which led to improvements in the presentation of our results. Additionally, we thank one of the referees for their thorough feedback and suggestions. 

This work was partially supported by QNRF grant NPRP 09-1150-2-448.

\bibliographystyle{model1-num-names}

\begin{thebibliography}{00}
\bibitem{CovDS}  {\sc M. Abdullah, C. Cooper, A. M. Frieze}, {\it Cover Time of a Random Graph with Given Degree Sequence},
Discrete Mathematics, Volume 312, Issue 21,  pp. 3146-3163, 2012 

\bibitem{AlFi} {\sc D. Aldous and J. Fill}, {\em Reversible Markov Chains and
Random Walks on Graphs}, (in preparation) \url{http://stat-www.berkeley.edu/pub/users/aldous/RWG/book.html}

\bibitem{Benezit} {\sc F. Benezit, P. Thiran, M. Vetterli} {\it The Distributed Multiple Voting Problem}, 
IEEE Journal of Selected Topics in Signal Processing, vol. 5, num. 4, p. 791-804, 2011

\bibitem{CFR-Mult}  {\sc C. Cooper, A. M. Frieze, T. Radzik}, {\it Multiple Random Walks in Random Regular Graphs}, 
SIAM Journal on Discrete Mathematics, Volume 23, Issue 4,  pp. 1738-1761, 2009 

\bibitem{ColinVoting} {\sc C. Cooper, R.  Els\"{a}sser, T. Radzik}, {\it The Power of Two Choices in Distributed Voting},
Proc. of The 41st International Colloquium on Automata, Languages, and Programming (ICALP), 2014 (to appear)

\bibitem{Cruise} {\sc J. Cruise,  A. Ganesh}, {\it Probabilistic Consensus via Polling and Majority Rules},
Proc. of Allerton Conference, 2010

\bibitem{Draief}  {\sc M. Draief, M. Vojnovic}, {\it Convergence Speed of Binary Interval Consensus}, 
SIAM Journal on Control and Optimization, 2012

\bibitem{ConcMeasure} {\sc D. P. Dubhashi, A. Panconesi}, {\em Concentration of Measure for the Analysis of Randomized Algorithms}, 
 Cambridge University Press, 2009

\bibitem{Peleg} {\sc Y. Hassin, D. Peleg}, {\it Distributed Probabilistic Polling and Applications to Proportionate Agreement},
Information and Computation 171, 248-268, 2001

\bibitem{Remco}{\sc R. van der Hofstad}, {\em Random Graphs and Complex Networks} \url{http://www.win.tue.nl/~rhofstad/NotesRGCN.pdf}, 2013

\bibitem{Janson} {\sc 	S. Janson, T. Luczak and A. Rucinski}, {\em Random Graphs}, Wiley, 2000

\bibitem{Montanari} {\sc Y. Kanoria,  A. Montanari}, {\it Majority Dynamics on Trees and the Dynamic Cavity Method}, 
Annals of Applied Probability, 2010


\bibitem{Mossel} {\sc E. Mossel , J. Neeman, O. Tamuz}, {\it Majority Dynamics and Aggregation of Information in Social Networks}
2012, arXiv:1207.0893 

\bibitem{Milan}  {\sc E.Perron, D. Vasudevan, M. Vojnovic}, {\it Using Three States for Binary Consensus on Complete Graphs},
IEEE Infocom 2009, IEEE Communications Society, 2009

\bibitem{Shang} {\sc S. Shang, P.W. Cuff, S.R. Kulkarni, P. Hui}, {\it An Upper Bound on the Convergence Time for Distributed Binary Consensus},  15th International Conference on Information Fusion (FUSION), 2012
\end{thebibliography}

\section{Appendix}
\begin{proof}[Proof of Proposition \ref{binprop}]
Let $X$ and $Y$ be independent random variables with distributions  $\operatorname{Bin}(2N,p)$ and $\operatorname{Bin}(2,p)$ respectively, and let $Z=X+Y$. Then $\mathbf{1}_{\{X \geq N\}}=\mathbf{1}_{\{Z \geq N+1\}}$ except when  $X=N$ and  $Y=0$,  or $X=N-1$ and $Y=2$. The former case occurs with probability $p_a=\binom{2N}{N}p^N(1-p)^{N+2}$ and the latter with $p_b=\binom{2N}{N-1}p^{N+1}(1-p)^{N+1}$. Observe $p_a \geq p_b$ if and only if $p \leq \frac{N+1}{2N+1}$, which is always the case when $p < \frac{1}{2}$. 
\end{proof}

\vspace{5mm}

\begin{theorem}[\cite{ConcMeasure}]\label{Chernoff}
Let $X=\sum_{i=1}^nX_i$ where the $X_i$'s are independent random variables distributed in $[0,1]$. For $\epsilon>0$, 
\begin{equation*}
\Pr(X<(1-\epsilon)\E[X]) \leq \exp\left(-\frac{\epsilon^2}{2}\E[X]\right).
\end{equation*}
\end{theorem}

\end{document}